\documentclass[lettersize,journal]{IEEEtran}
\usepackage{amsmath,amsfonts}
\usepackage{algorithm}
\usepackage{array}
\usepackage[caption=false,font=normalsize,labelfont=sf,textfont=sf]{subfig}
\usepackage{textcomp}
\usepackage{stfloats}
\usepackage{url}
\usepackage{verbatim}
\usepackage{graphicx}
\usepackage{cite}
\usepackage{xcolor}
\usepackage{hyperref}
\usepackage{longtable}
\usepackage{newtxmath}
\usepackage{outlines}
\usepackage{algpseudocode}
\usepackage{multirow}

\let\originalleft\left
\let\originalright\right
\renewcommand{\left}{\mathopen{}\mathclose\bgroup\originalleft}
\renewcommand{\right}{\aftergroup\egroup\originalright}

\newcommand{\bP}[2][]{\Pr\ifthenelse{\isempty{#1}}{}{_{#1}}\left[#2\right]}
\newcommand{\bE}[2][]{\mathop\mathbb{E}\ifthenelse{\isempty{#1}}{}{_{#1}}\left[#2\right]}
\newcommand{\bI}[2][]{\mathop\mathbb{I}\ifthenelse{\isempty{#1}}{}{_{#1}}\left[#2\right]}
\newcommand{\Var}[2][]{\mathbf{Var}\ifthenelse{\isempty{#1}}{}{_{#1}}\left[#2\right]}

\DeclareMathOperator*{\argmin}{arg\,min}

\newcommand{\tr}{{{\mathsf T}}}

\newcommand{\R}{\mathbb{R}}
\newcommand{\J}{\mathbf{J}}
\newcommand{\G}{\mathcal{G}}
\renewcommand{\d}{\mathrm{d}}

\newcommand{\revise}[1]{\textcolor{black}{#1}}

\usepackage{newtxmath}

\newtheorem{theorem}{Theorem}
\newtheorem{proof}{Proof}
\allowdisplaybreaks

\begin{document}

\title{Predicting Strategic Energy Storage Behaviors}

\author{Yuexin Bian~\IEEEmembership{Student Member,~IEEE}, Ningkun Zheng,~\IEEEmembership{Student Member,~IEEE}, Yang Zheng~\IEEEmembership{Member,~IEEE}, \\ Bolun Xu,~\IEEEmembership{Member,~IEEE}, Yuanyuan Shi~\IEEEmembership{Member,~IEEE}
\thanks{Yuexin Bian, Yang Zheng, and Yuanyuan Shi are with the Department of Electrical and Computer Engineering, University of California San Diego.}
\thanks{Ningkun Zheng, and Bolun Xu are with the Department of Earth and Environmental Engineering, Columbia University.}
}



\maketitle

\begin{abstract}
Energy storage are strategic participants in electricity markets to arbitrage price differences. Future power system operators must understand and predict strategic storage arbitrage behaviors for market power monitoring and capacity adequacy planning. This paper proposes a novel data-driven approach that incorporates prior model knowledge for predicting the strategic behaviors of \revise{price-taker} energy storage systems. We propose a gradient-descent method to find the storage model parameters given the historical price signals and observations. We prove that the identified model parameters will converge to the true user parameters under a class of quadratic objective and linear equality-constrained storage models. We demonstrate the effectiveness of our approach through numerical experiments with synthetic and real-world storage behavior data. The proposed approach significantly improves the accuracy of storage model identification and behavior forecasting compared to previous blackbox data-driven approaches. 
\end{abstract}


\section{Introduction}\label{Intro}


Energy storage is pivotal in balancing electricity demand and supply fluctuations in future decarbonized power systems, and tremendous investments have been made in both utility-scale and behind-the-meter (BTM) energy storage in recent years. The utility-scale energy storage capacity in the US has tripled in 2021, reaching 7.8 GW  storage as of Oct 2022, and is projected to reach 30~GW by 2025~\cite{us2021form, eia2022preliminary}. Additionally, the deployment rate of BTM energy storage is expected to exceed utility-scale energy storage, with the installed residential BTM energy storage capacity in the US increasing by 67\% year on year in 2022 Q2, mostly paired with distributed solar PV~\cite{fisher2017emissions, USESM, barbose2021behind}. Electricity markets are also removing barriers for energy storage participation, with Federal Energy Regulatory Commission (FERC) Order 841 ruling that all power system operators must allow storage to participate in all market services~\cite{federal2018electric} and FERC Order 2222~\cite{cartwright2020ferc} extending the requirement to BTM storage. 
As storage capacity continues to surge, price arbitrage in the energy market, buying electricity at a low price and selling it at a high price, is becoming one of the major revenue resources for energy storage owners to participate in electricity markets and earn revenue~\cite{us2021form}.

The deregulation of electricity markets and the incentive mechanisms for investing in renewable energy resources have made energy storage devices a strategic tool for profit in these markets. These devices, operated primarily by private entities, include both behind-the-meter (BTM) and utility-scale energy storage capacities~\cite{eia2020battery}. Typically, battery energy storage devices have a storage duration of one to four hours~\cite{zakeri2015electrical}, and their operators must carefully design their charging and discharge strategies to capture price fluctuations and maximize profits.

\revise{Energy storage devices play a critical role in enhancing power system flexibility. However, their strategic behavior can increase market volatility and undermine system robustness against load balance uncertainties~\cite{roozbehani2012volatility}. An accurate prediction of energy storage strategic behaviors is essential for market efficiency and to address concerns around market power~\cite{CAISO_ES_Proposal}. System operators can leverage the proposed algorithm for modeling the behavior of energy storage units and integrating them into the dispatch optimization process. Moreover, system operators can design new tariffs that align with the storage unit's for-profit interests with certain system-level objectives~\cite{fridgen2018one, zafar2018prosumer}, such as reducing peak demand power or reducing carbon emissions. In particular, a system operator can formulate the market tariff design as a bi-level optimization problem that optimizes for social welfare objectives while ensuring that the storage behavior aligns with the identified model. Last but not least, identifying the behavior models of storage market participants holds potential for market operators that are responsible for preventing market power abuse. One recent work~\cite{liang2023data} on this direction uses the
proposed approach for marginal offer price recovery of conventional generation units in the wholesale power market.}

\revise{This paper is an extended version of our earlier conference paper~\cite{bian2022demand}. We have made significant extensions by proposing sequential convex programming (SCP) to extend the gradient-based approach for identifying strategic energy storage behaviors with general objective functions, which is important for practical applications where utility functions are non-quadratic, or system operators have no prior knowledge about the parametric forms of agents' utility functions.} The main contributions of this work are as follows:
\begin{itemize}
    \item 
    \revise{We model strategic energy storage behaviors as a general agent decision-making optimization model. We then introduce a novel gradient-based approach for identifying the generic agent model, which can be used to forecast strategic energy storage behaviors accurately.}
    \item  \revise{We provide a formal convergence guarantee for convex quadratic battery agent models with linear equality constraints and discuss the local convergence property for problems with inequality constraints.}  
    \item \revise{We validate the proposed algorithm for predicting synthetic quadratic and generic energy storage behavior models and demonstrate its applicability on real-world datasets. Our approach accurately forecasts the strategic battery price arbitrage behaviors and outperforms simple heuristics, optimization-based methods, and competing machine-learning approaches with deep neural networks.}
\end{itemize}

We organize the remainder of the paper as follows. Section~\ref{LR} reviews related literature. Sections~\ref{Formu} and~\ref{Algo} present the problem formulation and solution algorithm, respectively. Section~\ref{CS} describes case studies with synthetic data and real-world applications. Section~\ref{Conclu} concludes this paper with future directions.

\section{Background and Related Work}\label{LR}



\subsection{\revise{Energy Storage Market Participation}}
\revise{Under FERC Order 841~\cite{federal2018electric}, energy storage can participate in the market through two options: self-scheduling and economic bids. On the one hand, Self-scheduling involves submitting offers indicating discharge and charge quantities for each time interval during the upcoming operating days. On the other hand, economic bids involve submitting both prices and quantities to express the willingness to discharge energy when the market price exceeds the proposed discharging price, or to charge energy when the market price falls below the proposed charging price. 
Economic bids offer greater efficiency compared to self-scheduling, especially in the presence of price uncertainty~\cite{mohsenian2015optimal, bhattacharjee2022energy}.}

\revise{Behind-the-meter (BTM) energy storage participation refers to the integration of energy storage systems (ESSs) on the customer side of the electric meter, typically at residential, commercial, or industrial buildings~\cite{rezaeimozafar2022review}. Due to their nature of not being directly cleared in the wholesale market, BTM energy storage systems often utilize a price response strategy instead of participating by self-scheduling or economic bidding. Using price response strategy, BTM energy storage owners can make charging and discharging decisions based on real-time price signals, rather than submitting bids in advance. The adoption of BTM energy storage is aimed at providing customers with benefits such as electricity bill savings and demand-side management~\cite{irena2019behind}. In recent years, there has been a significant increase in the installed capacity of BTM energy storage systems, coinciding with the growth in distributed generators~\cite{mcilwaine2021state}.}

\revise{Therefore, both economic bids and behind-the-meter (BTM) energy storage systems make price-responsive decisions regarding charging and discharging activities. In our work, we model these strategic energy storage behaviors
as price-takers, with the goal of faithfully capturing the characteristics of real-world energy storage participation behavior.}
\revise{\subsection{Market Power Mitigation}}
\revise{With the surging installed storage capacity and the emergence of strategic behaviors, the need to alleviate utility-scale energy storage market power has become a new challenge in system design. 
While there has been extensive research on market power mitigation for conventional generators, which relies primarily on fuel cost and heat rate curve parameters to determine the optimal bidding strategy, there is a growing recognition of the need to consider market power mitigation for energy storage systems as well~\cite{caiso2022}. The two primary approaches to market power mitigation in the US, the structural approach and the conduct and impact approach, both require market power monitoring and  marginal cost estimation of generation resources~\cite{graf2021market}. 
Energy storage systems differ from conventional generators as they need to account for the opportunity costs associated with charging and discharging. 
Previous works have discussed the effect of energy storage market power on storage owner profit and social welfare~\cite{mohsenian2015coordinated,contereras2017energy}. Therefore, it is important to develop market power mitigation methods that consider the unique characteristics of energy storage systems.}

\revise{
This paper can provide a fundamental basis for market regulation design by identifying the strategic behavior of utility-scale energy storage systems. Specifically, we assume that energy storage agents participate in electricity markets with the aim of minimizing their costs while satisfying the system constraints. To achieve this goal, we propose a novel approach based on historical market participation data to identify storage agent models and associated parameters, including its marginal operation cost (i.e., degradation cost), power/energy capacity limits, and efficiency. 
By identifying the marginal operation cost, it helps facilitate market behavior monitoring and help identify instances of market power abuse. One recent work~\cite{liang2023data} on this direction is using the proposed approach for  marginal offer price recovery of conventional generation units in wholesale power markets.}

\revise{\subsection{Price-responsive Demand Response Behaviors Forecasting}}
Our paper focus on price-responsive behavior. Existed methods can be separated into two categories. The first category lies in model-free data-driven approaches. For instance, ~\cite{nghiem2017data} introduced a Gaussian model to predict the demand response of buildings in response to price signals and~\cite{vuelvas2019novel} used a two-layer neural network to characteristic the consumer behaviors under an incentive-based demand response program.
The second category utilizes the agent model knowledge, which aims to predict the energy decision using the assumed model. These works typically employ inverse optimization~\cite{ fernandez2021forecasting,kovacs2021inverse,lu2018data}. The inverse optimization formulates the agent identification problem as a bi-level optimization problem to find the best model parameters using the price and user behavior data. The upper-level problem minimizes the mean absolute error of the predicted and actual user behavior, and the lower-level problem models the user decision-making model for strategic behavior prediction. %

\revise{Optimization-based methods incorporate agent model knowledge and achieve better prediction accuracy than black-box data-driven models with limited data. However, data-driven methods are easier to adapt to different applications, while optimization-based methods can struggle with large datasets and are not robust to noisy data. To leverage the benefits of both approaches, this paper focuses on combining the model-based and model-free methods for joint battery model identification and behavior prediction. }

\begin{figure*}
    \centering
    \includegraphics[trim = 0mm 0mm 0mm 0mm, clip, width=1.8\columnwidth]{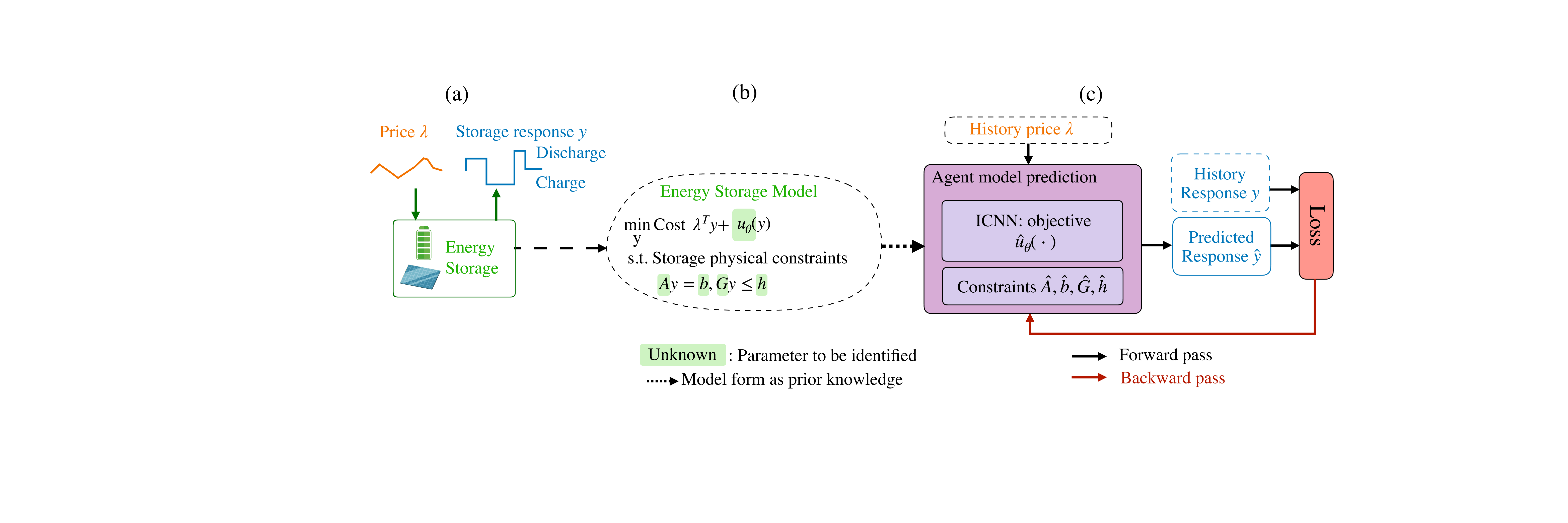}
    \caption{Diagram of the proposed energy storage agent model identification and forecasting framework. Prior knowledge of the energy storage agent is modeled as an optimization problem, in which the objective is to minimize the energy cost and degradation cost, subject to storage physical constraints. Parameters in the energy storage models are unknown to the system operator. We use a gradient-based method to update and identify the parameters in the energy storage model, to minimize the difference between the predicted storage response and the actual response.} 
    \label{fig:frame}
\end{figure*}

\section{Problem Formulation}\label{Formu}



\revise{Fig. \ref{fig:frame} provides an overview of the proposed model-based energy storage behavior forecast approach. We model strategic energy storage behaviors as a general agent decision-making optimization model. Specifically, we assume each energy storage participant conducts private optimization to determine their storage response, with the goal of minimizing the operation cost (charging electricity cost plus storage degradation cost) while satisfying physical constraints. 
The model is trained on minimizing the difference between the predicted storage response and the observed storage response on the historical data, which returns the identified storage utility function and physical parameters as by-products of the prediction model.}

\subsection{Energy Storage Demand Response Model}\label{sec:energy_model}
We start by considering a price-responsive energy storage agent model subject to a time-varying electricity price $\lambda \in \mathbb{R}^T$ and a disutility cost function,
\begin{subequations}\label{eq:agent_battery}
\begin{align}
  \min_{p,d}\quad & \sum_{t=1}^{T} \lambda_t (p_t-d_t) + u_\theta(p,d) \label{eq:agent_obj}\\
  \text{subject to} \quad & \underline{P} \leq d_t, p_t \leq \overline{P}\,, \label{eq:agent_battery_contr1}\\
  \quad & e_t = e_0 + \sum_{\tau=1}^t p_{\tau}\eta_c - \frac{d_\tau}{\eta_d}, \label{eq:agent_battery_contr2}\\
  \quad & 
  \underline{E_m} \leq e_t \leq \overline{E_m}\,,  \label{eq:agent_battery_contr3}
\end{align}
\end{subequations}
where decision variables $p_t, d_t$ are the energy storage charging and discharging power, $e_0, e_t$ is the relative state of charge level at initial stage and operating time $t$, $u_\theta(\cdot)$ models the storage disutility cost, which is a generic function of the storage charging/discharging power. Here $\underline{P}, \overline{P}, \underline{E_m}, \overline{E_m}$ are the maximum and minimum storage power and energy constraints, and $\eta_c, \eta_d$ is the charging efficiency and discharging efficiency. \revise{The expression $\sum_{\tau=1}^t p_{\tau}\eta_c - \frac{d_\tau}{\eta_d}$ represents the accumulated relative state of charge level from time 1 to time $t$}. We assume $e_0$ is known. \revise{Here parameters of each energy storage agent are $\Theta = \{\theta, \underline{P}, \overline{P}, \underline{E_m}, \overline{E_m}, \eta_c, \eta_d\}$, where $\theta$ are parameters in the agent's disutility function $u_{\theta}(\cdot)$.}
We re-write \eqref{eq:agent_battery_contr2} and \eqref{eq:agent_battery_contr3} as,
\begin{equation}\label{eq:storage_contr2_reformulation}
    \underline{E_m}\eta_d \leq e_0 \eta_d + \sum_{\tau=1}^{t} p_{\tau}\eta_c\eta_d - d_{\tau} \leq \overline{E_m}\eta_d,
\end{equation}
where $\eta_c\eta_d, {(\underline{E_m} -e_0) \eta_d, (\overline{E_m} - e_0)\eta_d}$ are the new parameters. \eqref{eq:storage_contr2_reformulation} is a linear constraint and we can recover the original parameter values ${\underline{E_m}, \overline{E_m}, \eta}$ after the learning process if charging efficient and discharging efficiencies are symmetric, i.e., $\eta_c=\eta_d = \eta$, or we can model the round-trip efficiency $\eta_c \eta_d$ as one parameter rather than modeling the charging and discharging efficiency separately.



To be more concrete, let us consider the following example within this framework.

\noindent\emph{Example 1: Price Arbitrage with Quadratic Storage Degradation Cost.} Here we consider $u_\theta$ to be parametric, where cost function is a quadratic function.
\begin{subequations}\label{eq:agent_battery1}
\begin{align}
  \min_{p,d}\quad & \sum_{t=1}^{T} \lambda_t (p_t-d_t) + c_1d_t + c_2d_t^2 \label{eq:quad_battery_obj1}\\
  \text{subject to} \quad & \eqref{eq:agent_battery_contr1}, \eqref{eq:storage_contr2_reformulation}
\end{align}
\end{subequations}

\noindent\emph{Example 2: Price Arbitrage with SoC-dependent Storage Degradation model}.
Here we consider $u_{\theta}$ to be an SoC-related storage degradation cost~\eqref{eq:agent_battery2}, which includes a quadratic degradation cost with respect to the charging power and an SoC-cost related to the SoC level.
The SoC-cost is lower when operating around $50\%$ SoC level, and higher when SoC approaches $0\%$ or $100\%$. 
\begin{subequations}\label{eq:agent_battery2}
\begin{align}
  \min_{p,d}\quad & \sum_{t=1}^{T} \lambda_t (p_t-d_t) + c_1(e_{t}-0.5)^2 + c_2d_t^2 \label{eq:soc_battery_obj1}\\
  \text{subject to} \quad & ~\eqref{eq:agent_battery_contr1}, \eqref{eq:storage_contr2_reformulation}
\end{align}
\end{subequations}
The objective functions in both examples, e.g., \eqref{eq:quad_battery_obj1} and \eqref{eq:soc_battery_obj1} can be captured by the general formulation in \eqref{eq:agent_obj}, subject to the physical constraints of charging/discharging power and energy limits. 
\revise{\textbf{Remark 1.} In this paper we focus on predicting energy storage price-responsive behaviors with the agent behavior model~\eqref{eq:agent_battery} specifies the state of the charge dynamics in \eqref{eq:agent_battery_contr2}, and maximum power/energy capacity \eqref{eq:agent_battery_contr1}-\eqref{eq:agent_battery_contr3} for storage units. It is worth noting that the proposed algorithm can be adapted to other price-responsive demand response agent modeling and forecast problems with affine constraints and generic convex cost functions, embedding prior knowledge about the specific applications via certain parametric objective/constraint forms. }

\subsection{Energy Storage Model Identification}
Here we formulate the energy storage model identification as an optimization problem.
We consider the energy storage agent conducts a \emph{private} optimization \eqref{eq:agent_battery} to decide its charging and discharging power $p$ and $d$, with respect to a price signal $\lambda$. The true disutility function and parameters in the agent decision model, i.e., $\bar{\Theta} = \{\theta, \underline{P}, \overline{P}, \underline{E_m}, \overline{E_m}, \eta_c, \eta_d\}$ are unknown to the system operator. Instead, the system operator only has access to a collection of $N$ historical price and agent response data pairs. 
The goal is to learn the agent behavior model (i.e., the parameter $\Theta$) that minimizes the difference between the predicted and the actual storage response, 
\begin{equation}\label{eq:identification}
\begin{aligned}
    \min_{\Theta}\quad &L := \frac{1}{N} \sum_{i=1}^N \left(\|p_{i}^\star-p_{i}\|^2 + \|d_{i}^\star-d_{i}\|^2\right)\\
    \text{where}\quad & p_i^\star, d_i^\star \in \arg\min\; \revise{\text{total storage cost}~\eqref{eq:agent_obj} }\\
    &\qquad\quad\revise{ \text{s.t. storage agent constraints}~\eqref{eq:agent_battery_contr1}, \eqref{eq:storage_contr2_reformulation}}
\end{aligned}
\end{equation}
in which the training~/~identification inputs include
\begin{itemize}
    \item $\lambda_{i}$ is the price signal for sample scenario $i$;
    \item $p_i, d_i$ is the actual storage response;
\end{itemize}
and the problem has the following variables
\begin{itemize}
    \item $\Theta$ represent learnable parameters in the storage agent model; 
    \item $p_i^\star, d_i^\star$ is the predicted storage response based on estimated model $\Theta$. 
\end{itemize}

\section{Algorithm Design}\label{Algo}

We propose a gradient-descent approach to solve the generic energy storage agent identification problem~\eqref{eq:identification}. We start by presenting the gradient-descent training algorithm. Then we discuss how to derive the gradients of the optimization problem where the objective of the energy storage agent is in the quadratic function form; and discuss how we extend the method when the agent model is a generic one. Finally, we provide a convergence analysis for the gradient-based agent identification algorithm under the convex quadratic energy storage agent models with only linear equality constraints.

\subsection{Gradient-based Approach for Storage Model Identification}\label{sec:gradientformodel}
To solve the energy storage agent identification problem in \eqref{eq:identification}, we propose a gradient-descent approach.  
Without loss of generalization, we can present the storage model with a generic objective function and affine constraints. 
\begin{subequations}
\label{eq:agent_model2}
\begin{align}
   \min_{y}\quad &c(y, \lambda) :=  \lambda^\top y + u_\theta(y) \label{eq:agent2_obj}\\
   \text{subject to} \quad & Ay = b \,, \label{eq:agent2_contr1}\\
   & Gy \leq h \,,\label{eq:agent2_contr2}
\end{align}
\end{subequations}
where $\lambda = \begin{bmatrix}\lambda_1, \ldots, \lambda_T, -\lambda_1, \ldots, -\lambda_T \end{bmatrix}^\top \in \mathbb{R}^{2T}$, $\lambda_t$ denoting the time-of-use price signal at time step $t$, $y = \begin{bmatrix}p_1, \ldots, p_T, d_1, \ldots, d_T\end{bmatrix}^\top \in \mathbb{R}^{2T}$ is the decision variable denoting the storage response to the price signal. Matrices $A \in \mathbb{R}^{m\times 2T}, b \in \mathbb{R}^m, G \in \mathbb{R}^{p \times 2T}, h\in \mathbb{R}^{p}$ define the collection of equality and inequality constraints in the agent model. $u_\theta(\cdot)$ models the storage disutility cost. We want to learn $\Theta = \{\theta, A, b, G, h\}.$
Then loss in~\eqref{eq:identification} can be written as
\begin{equation}
    L := \frac{1}{N}\|y_i^\star - y_i\|^2
\end{equation}
At iteration $t$, 
\begin{equation}\label{weight_f2}
    \Theta^{(t+1)} \leftarrow \Theta^{(t)} - \eta \sum_{i=1}^{N} \frac{\partial L}{\partial y_i^\star}
    \frac{\partial y_i^\star}{\partial \Theta} \,, 
\end{equation}
where $\eta$ is the learning rate; 
$\frac{\partial L}{\partial y_i^\star}$ is the gradient of the loss function with respect to the storage response forecast; and $\frac{\partial y_i^\star}{\partial \Theta}$ is the gradient of the storage response forecast with respect to the optimization problem parameters, evaluated at $\Theta^{(t)}$. 

Computing the derivative of the loss function $L$ with respect to storage response $y^\star$, i.e., $\frac{\partial L}{\partial y^\star}$ is trivial. The main challenge is in computing the gradient of the optimal solution $y^{\star}$ with respect to the problem parameters $\Theta$, i.e., $\frac{\partial y^\star}{\partial \Theta}$, which will be detailed in the next two subsections.

\revise{\textbf{Remark 2.} In our paper, we use a gradient descent approach over other descent-based techniques, such as Newton or Quasi-Newton, primarily because deriving the first-order gradient involves differentiating KKT conditions, which requires solving a costly Jacobian-based equation. It is more computationally expensive to compute the Hessian matrix and its inverse needed for second-order methods. However, we acknowledge that it could be an interesting direction for future work to investigate the use of second-order optimization methods.}
\subsection{Special Case: Differentiating Quadratic Storage Models}
\label{subsect:Storage_optnet}

To begin with, let us consider a quadratic energy storage agent model~\eqref{eq:quad_agent_model}, where the objective function $u_{\theta}(y) = q^\top y + y^\top Q y$ thus the parameters to be identified can be instantiated as $\Theta = \{q, Q, A, b, G, h\}$.
\begin{subequations}
\label{eq:quad_agent_model}
\begin{align}
   \min_{y}\quad &c(y, \lambda) :=  \lambda^\top y + q^\top y + y^\top Q y \label{eq:quad_agent_obj}\\
   \text{subject to} \quad & Ay = b \,, \label{eq:quad_agent_contr1}\\
   & Gy \leq h \,,\label{eq:quad_agent_contr2}
\end{align}
\end{subequations}
First, we claim that the strong duality holds if our energy storage agent model~(\ref{eq:quad_agent_model}) is strictly feasible. The strong duality holds if the problem is a convex problem and strictly feasible~\cite{boyd2004convex}. The  agent model~\eqref{eq:quad_agent_model} is convex because it has  a convex objective and affine constraints. The problem models the storage response, where $\mathrm{domain}(y)=\mathbb{R}^{2T}$. Under the storage identification and forecast setting, it is reasonable to assume that the energy storage agent model is strictly feasible, that is, there exists 
$y \in \mathbb{R}^{2T}$ 
such that $Ay=b,Gy < h$. Then, strong duality holds for~\eqref{eq:quad_agent_model} since Slater's condition holds. 

Since strong duality holds for \eqref{eq:quad_agent_model}, we write out its Karush–Kuhn–Tucker (KKT) conditions, which are sufficient and necessary for optimality, 
\begin{subequations}
\begin{align}
    \lambda + q +  Q y^\star  
    + A^\top\nu^\star + G^\top\mu^\star = 0 \\
    Ay^\star -b = 0 \\
    D(\mu^\star)(Gy^\star - h) = 0,
\end{align}
\end{subequations}
where $D(\cdot)$ creates a diagonal matrix from a vector, $\lambda$ is a vector $\lambda = [\lambda_1, \ldots, \lambda_T, -\lambda_1, \ldots, -\lambda_T]^\tr \in \R^{2T}$, $y^\star = [y^\star_1, y^\star_2, \ldots, y^\star_{2T}]^\tr \in \R^{2T}$ is the optimal primal solution, and $\nu^\star, \mu^\star$ are the optimal dual variables on the equality constrains and inequality constraints. 

Following the method in \cite{bian2022demand},  we take the total derivatives of these conditions and put them in a compact matrix form,
\begin{equation}
\begin{aligned}
&\underbrace{\begin{bmatrix}
Q  & G^\top & A^\top \\
D(\mu^\star)G & D(Gy^\star-h) & 0 \\
A & 0 & 0
\end{bmatrix}}_{K}
\begin{bmatrix}
\d y \\ \d \mu \\ \d \nu
\end{bmatrix}
= 
\\ &\qquad \qquad -
\begin{bmatrix}
\d Q y^\star + \d \lambda + \d q + \d G^\top \mu^\star + \d A^\top \nu^\star \\
D(\mu^\star)\d Gy^\star - D(\mu^\star)\d h \\
\d Ay^\star - \d b
\end{bmatrix},
\end{aligned}
\label{eq:kkt}
\end{equation}
and we can obtain 
the derivatives $\frac{\partial y^*}{\partial Q}$, $\frac{\partial y^*}{\partial q}$, $\frac{\partial y^*}{\partial b}$, $\frac{\partial y^*}{\partial h}$, $\frac{\partial y^*}{\partial A}, \frac{\partial y^*}{\partial G}$,
\begin{subequations}\label{eq:linear_system}
\begin{align}
    &\begin{bmatrix}\frac{\partial y^*}{\partial Q} & \frac{\partial \mu^*}{\partial Q} & \frac{\partial \nu^*}{\partial Q}\end{bmatrix}^\tr = -K^{-1}\begin{bmatrix}
    \frac{\d Q}{\d Q}y^\star & 0 & 0 
    \end{bmatrix}^\tr,\\
    &\begin{bmatrix}\frac{\partial y^*}{\partial q} & \frac{\partial \mu^*}{\partial q} & \frac{\partial \nu^*}{\partial q}\end{bmatrix}^\tr = -K^{-1}\begin{bmatrix}I & 0 & 0\end{bmatrix}^\tr, \\ 
    &\begin{bmatrix}\frac{\partial y^*}{\partial h} & \frac{\partial \mu^*}{\partial h} & \frac{\partial \nu^*}{\partial h}\end{bmatrix}^\tr = -K^{-1}\begin{bmatrix}0 & -D(\mu^\star) & 0\end{bmatrix}^\tr,  \\
    &\begin{bmatrix}\frac{\partial y^*}{\partial b} & \frac{\partial \mu^*}{\partial b} & \frac{\partial \nu^*}{\partial b}\end{bmatrix}^\tr = -K^{-1}\begin{bmatrix}0 & 0 & -I\end{bmatrix}^\tr, \\
    &\begin{bmatrix}\frac{\partial y^*}{\partial A} & \frac{\partial \mu^*}{\partial A} & \frac{\partial \nu^*}{\partial A}\end{bmatrix}^\tr = -K^{-1}\begin{bmatrix}\frac{\d A^T}{\d A}\nu^\star & 0 & \frac{\d A}{\d A}y^\star\end{bmatrix}^\tr, \\
    &\begin{bmatrix}\frac{\partial y^*}{\partial G} & \frac{\partial \mu^*}{\partial G} & \frac{\partial \nu^*}{\partial G}\end{bmatrix}^\tr = -K^{-1}\begin{bmatrix}\frac{\d G^T}{\d G}\mu^\star & y^{\star T} \otimes D(\mu^\star)& 0\end{bmatrix}^\tr,
\end{align}
\end{subequations}
where $\frac{\d Q}{\d Q} \in \R^{2T \times 2T \times 2T \times 2T}$, $\frac{\d A^T}{\d A} \in \R^{2T \times m \times 2T \times m}$, $\frac{\d A}{\d A} \in \R^{m \times 2T \times m \times 2T}, \frac{\d G^T}{\d G} \in \R^{2T \times p \times p \times 2T}$.
Note that the derivative of a parameter with respect to itself (e.g., $\frac{\d q}{\d q}$) is an identity matrix while the derivative to different parameters (for example, $\frac{\d q}{\d b}$) is zero.

The training process is summarized in Algorithm~\ref{alg:train}. For a given training dataset $\{\lambda_{i}, y_i\}$, $\lambda_i$ is the processed price signal $\lambda_i \in \R^{2T}$ and $y_i$ is observed storage response $y_i \in \R^{2T}$ for sample $i$. In the forward pass, we can get the estimated energy storage agent decision by we obtain $y_i^\star = \argmin_y c(y; \lambda_i, \Theta^{(t)})$ where $\Theta^{(t)}$ is the model parameter estimation at iteration $t$, subsequently, compute the loss function $L = \frac{1}{N} \sum_{i=1}^{N} \left\|y_i^\star-y_i\right\|^2$. In the backward pass, we need to compute the derivative $\frac{\partial L}{\partial \Theta}$ and update the model estimate $\Theta$ via gradient descent.
\begin{algorithm}[t]
\caption{Quadratic Storage Model Identification}
\label{alg:train}
\begin{algorithmic}
\Require Dataset $D = \{\lambda_i, y_i\}, i=1, \ldots, N$.
\Ensure $\theta^{(0)}$ \Comment{initial parameter of storage model}
\For{$ite = 0, \ldots, T-1$}
\State \textbf{sample} batch $B = \{(\lambda_i, y_i), i = 1, \ldots, |B|\}$ from $D$
\State \textbf{compute} $y_i^\star = \argmin_y c(y;\lambda_i,\theta^{(t)}), i = 1,\ldots, |B|$
\State \textbf{Evaluate} the loss function $L = \frac{1}{|B|} \sum_{i=1}^{|B|} \left\|y_i^\star-y_i\right\|^2$
\State \textbf{update} $\theta^{(ite)}$ with $\nabla_{\theta^{(ite)}} L$ using ~\eqref{weight_f2}, \eqref{eq:linear_system}:
$$\theta^{(ite+1)} \leftarrow \theta^{(ite)} - \eta \sum_{i=1}^{|B|} \frac{\partial L}{\partial y_i^\star}  \frac{\partial y_i^\star}{\partial \theta}$$
\EndFor
\end{algorithmic}
\end{algorithm}

\begin{algorithm}[t]
\caption{Generic Storage Model Identification}\label{alg:train2}
\begin{algorithmic}
\Require Dataset $D = \{\lambda_i, y_i\}, i=1, \ldots, N$.
\Ensure $\Theta^{(0)}$ \Comment{initial parameter of storage model}
\For{$ite = 0, \ldots, T-1$}
\State \textbf{sample} batch $B = \{(\lambda_i, y_i), i = 1, \ldots, |B|\}$ from $D$
\State \textbf{form the convex approximation of storage model}
\State Initialize $y^{(1)}$
\For{k = 1 to [converged]}
\State Approximate $\hat{u}_{\theta}^{(k)}$ around $y_i^{(k)}, \forall i \in B$
\State Obtain the solution $y_i^{(k+1)}, \forall i \in B$ by solving the approximated convex quadratic storage Storage problem~\eqref{eq:agent_model_approximate}
\EndFor
\State \textbf{compute} $y_i^\star$ by solving the final convex approximation $\forall i \in B$
\State \textbf{evaluate} the loss function $L = \frac{1}{|B|} \sum_{i=1}^{|B|} \left\|y_i^\star-y_i\right\|^2$
\State \textbf{update} $\Theta^{(ite)}$ with $\nabla_{\Theta^{(ite)}} L$ using ~\eqref{weight_f2}, \eqref{eq:linear_system} and \eqref{eq:derivative2}:
$$\Theta^{(ite+1)} \leftarrow \Theta^{(ite)} - \eta \sum_{i=1}^{|B|} \frac{\partial L}{\partial y_i^\star} \frac{\partial y_i^\star}{\partial \Theta}$$
\EndFor
\end{algorithmic}
\end{algorithm}

\subsection{Differentiating Generic Storage Models}
While the quadratic storage model is a powerful model in capturing the degradation cost~\cite{trippe2014charging}\cite{tan2017optimal}, in practice, energy storage agents might have more complex degradation costs~\cite{shi2018convex} and objective functions~\cite{xu2018optimal}. Furthermore, if the system operator does not know the exact storage utility function parametric form, using the quadratic storage model for identification might lead to a suboptimal solution due to potential model mismatches. 
In this section, we extend our method to more generic objective functions in \eqref{eq:agent_model2},
\begin{subequations}
\begin{align*}
    \min_{y}\; & \lambda^\top y  + u_{\theta}(y) \,,\\
    \text{s.t. } & Ay = b\,, Gy \leq h.
\end{align*}
\end{subequations}

\label{sec:scp}
A common approach to solve problem~\eqref{eq:agent_model2} is sequential convex programming~\cite{duchi2018}. Sequential convex programming is a local method for solving nonconvex problems that leverage convex optimization. The basic idea is to \emph{iteratively} form a convex approximation problem and optimize over the convex approximation problem~\eqref{eq:agent_model_approximate},
\begin{subequations}\label{eq:agent_model_approximate}
\begin{align}
    y^{(i)} = \min_{y}\; & \lambda^\top y  + \hat{u}^{(i)}_{\theta}(y) \,, \\
    \text{s.t. } &Ay = b, Gh \leq h, 
\end{align}
\end{subequations}
until convergence.
Here we approximate the cost function $u_\theta(\cdot)$ using the second-order Taylor approximation around the $i$-th solution $y^{(i)}$, 
\begin{align}\label{eq:approx_obj}
    \hat{u}^{(i)}_{\theta}(y) = u^{(i)}_{\theta}(y^{(i)}) + {q^{(i)}}^\top(y-y^{(i)}) 
     + (y-y^{(i)})^\top Q^{(i)}(y-y^{(i)}), 
\end{align}
where $q^{(i)} = \nabla_{y^{(i)}} u \in \mathbb{R}^{2T}$ and $Q^{(i)} = \nabla^2_{y^{(i)}} u \in \mathbb{R}^{2T \times 2T}$. 

Suppose we achieve the fixed point $y^\star = y^{(k)}$ at $k$-th iteration, with the approximated quadratic storage objective function and affine constraints, differentiating the solution with respect to the unknown agent parameters can be done as we show in subsection~\ref{subsect:Storage_optnet}.
The gradient of storage response forecast $y^\star$ with respect to parameters $A,b,G,h$ in affine constraints can be directly derived following~\eqref{eq:linear_system} in Section~\ref{subsect:Storage_optnet}. To obtain the gradient of storage response forecast with respect to the unknown parameters in the objective function, $\frac{\partial y^\star}{\partial \theta}$, we use the chain rule,
\begin{equation}\label{eq:derivative2}
    \frac{\partial y^\star}{\partial \theta} = \frac{\partial y^\star}{\partial \hat{q}}\frac{\partial \hat{q}}{\partial \theta} + \frac{\partial y^\star}{\partial \hat{Q}}\frac{\partial \hat{Q}}{\partial \theta}\,\,, 
\end{equation}
where $\hat{q} = \nabla_{y^\star} u \in \mathbb{R}^{2T}$, $\hat{Q} = \nabla^2_{y^\star} u \in \mathbb{R}^{2T \times 2T}$.


The training process for general energy storage agent models are summarized in Algorithm~\ref{alg:train2}. Specially, we model the unknown objective function $u_{\theta}$ via the input convex neural networks (ICNN)~\cite{amos2017icnn}, which guarantees the second order approximation $\hat{Q} \in \mathbb{R}^{2T \times 2T}$ is always a positive semi-definite matrix. Thus, the resulting second-order Taylor approximation $\hat{u}_\theta^{(i)}(y)$ \eqref{eq:approx_obj} is always a convex function.
In the forward pass, given $\theta^{(ite)}$ is the current neural network parameters at iteration $ite$, we iteratively form the convex approximation problem~\eqref{eq:agent_model_approximate} and optimize over the corresponding convex problem, to obtain the fixed point $y^{*}_{i}$ for training sample $i$.
In the backward pass, we can get derivatives $\frac{\partial y^\star}{\partial \Theta}$ by differentiating the convergent convex program~\eqref{eq:derivative2} and then updating parameters accordingly. 
\revise{In practice, directly solving the approximate convex problem has been found to work well and is often computationally efficient~\cite{duchi2018}. However, we acknowledge that there can be certain situations or applications where the use of convex approximations may not be the most suitable choice due to specific constraints or requirements. In the experimental section, we demonstrate that sequential convex approximation consistently converges within 20 iterations.}





\subsection{Convergence Analysis}
For a subset of quadratic storage models, the gradient-based approach (Algorithm \ref{alg:train}) has convergence guarantee.
Specially, consider the agent model only has equality constraints and the objective function with only the quadratic term $Q = \alpha I$ (unknown) and $q = 0$. In addition, we assume $A$ is known and focus on identifying the constraint parameters $b$. 
We first prove convergence of the agent model with equality constraints and discuss the local convergence for the inequality case.
The loss function in \eqref{eq:identification} can be concretely written as,
\begin{subequations}\label{eq:qp_equ_loss}
\begin{align}
    L(\alpha, b) &= \frac{1}{N} \sum_{i=1}^{N} ||\hat{y}_i - y_i ||_2^{2}\,, \\
    \text{where}  \quad  & \hat{y}_i = \argmin \;\lambda_i^\tr y + \frac{\alpha}{2} y^\tr y \,, \\
    & \quad \text{subject to} \; \, Ay = b \quad : (\nu_i)
\end{align}
\end{subequations}
and $\hat{y}_i$ is the solution to the equality-constrained problem, 

\begin{equation}\label{eq:sol_equality}
    \hat{y}_i = \frac{1}{\alpha} \left(A^\top (AA^\top)^{-1} A - I\right) \lambda_i + A^\top (AA^\top)^{-1} b.
\end{equation}

Thus, the loss function has the following form,
\begin{equation}\label{eq:loss_eq}
    L(\alpha, b) = \frac{1}{N}\sum_{i=1}^N \left\|\frac{1}{\alpha} k_{1i} + k_2 b - y_i\right\|^2 \,,
\end{equation}
where $k_{1i} = \left(A^\top (AA^\top)^{-1} A - I\right) \lambda_i$ and $k_2 = A^\top (AA^\top)^{-1}$. 
Theorem \ref{convergent_for_eq} shows that $L$ has a gradient dominance property w.r.t $\alpha, b$ which ensures the convergence of gradient update. 

\begin{theorem}\cite{bian2022demand}\label{convergent_for_eq}
For the equality-constrained quadratic agent model identification in~\eqref{eq:qp_equ_loss}, assume $\alpha > \delta > 0$, 
and for any initial value $\alpha_0, b_0 \in \mathbb{R}$, define a sublevel set $\G_{10\epsilon^{-1}} = \{\alpha, b \in \R | L(\alpha,b) - L(\alpha^\star,b^\star) \leq 10\epsilon^{-1}\Delta_0\}$, $\Delta_0 := L(\alpha_0,b_0) - L(\alpha^\star,b^\star)$ and $\epsilon > 0$ is any positive constant. $L(\alpha,b)$ has gradient dominance property and $\nabla L(\alpha,b)$ is Lipschitz bounded over $\G_{10\epsilon^{-1}}$. Therefore, using the gradient update $\alpha^{(t+1)} \leftarrow \alpha^{(t)} - \eta_\alpha \frac{\partial L}{\partial \alpha}\,, b^{(t+1)} \leftarrow b^{(t)} - \eta_b \frac{\partial L}{\partial b}$ guarantees,
$$\lim_{t \rightarrow \infty} \alpha^{(t)} \rightarrow \alpha^\star\,, \lim_{t \rightarrow \infty} b^{(t)} \rightarrow b^\star\,,$$ 
where $\alpha^\star, b^\star$ is the true model parameter that generates the training data $(\lambda_i, y_i), i=1, ..., N$.
\end{theorem}
Proof of Theorem 1 can be found in Appendix~\ref{sec:proof} for completeness.  
\revise{To the best of our knowledge, convergence analysis of differentiable optimization problems is still an open problem. The convergence results in the equality-constrained case (Theorem~\ref{convergent_for_eq}) provide a first step in understanding this question. 
For inequality-constrained quadratic agent model
\begin{equation}
    \hat{y}_i  = \argmin \lambda_i^\tr y + \frac{\alpha}{2}y^\tr y, \,\text{subject to}\,\, Gy \leq h
\end{equation}
The standard perturbation theory~\cite{robinson1974perturbed} demonstrates that if $\alpha_0$ and $h_0$ are sufficiently close to $a$ and $h$, respectively, the active inequality constraints can be identified. This result can transform the general inequality-constrained problems into equality-constrained problems, for which we provide a local convergence guarantee. As previously shown in \cite{bian2022demand}, globally the loss function can contain multiple local minima, even for $h \in \mathbb{R}$ (e.g. when there is only one inequality constraint). Several approaches have been proposed in the literature to escape local optima, such as initializing with different points, stochastic gradient descent~\cite{kleinberg2018alternative}, and adding perturbations~\cite{feng2020dynamical}. 
To address this issue, we employ a random shooting approach in our experiments. We train the model with different initial parameters and choose the set with the lowest training loss, which yields good prediction performance.
}

\section{Case Study}\label{CS}
We consider three applications of the proposed framework: 1) storage model identification and response prediction with a quadratic degradation cost function, 2) storage model identification and response prediction with an SoC-dependent degradation cost function, 3) model identification and response prediction on a real-world price arbitrage storage. 
We compare our approach with existing data-driven approaches including feedforward neural networks and recurrent neural networks, and show our storage response behavior prediction significantly outperforms the existing methods, in both the simulated energy storage models and the real-world dataset.
Source code, input data and the trained models for all experiments are available on Github\footnote{\url{https://github.com/alwaysbyx/Predicting-Strategic-Energy-Storage-Behaviors}}.

\subsection{Quadratic Energy Storage Model}
\textbf{Experiment setting.} 
We start with storage arbitrage with a quadratic disutility function following the model~\eqref{eq:agent_battery1}. 
We consider a storage model with a power rating $\overline{P} = 0.5MW, \underline{P} = 0MW$, starting with 50\% SoC, and vary the energy storage duration $H$ from $\{1h, 2h, 3h, 4h\}$ to calculate $\{\underline{E_m},\overline{E_m},e_0\}$, i.e., when $H = 1h, \underline{E_m} = 0 MWh, \overline{E_m} = 0.5 MWh, e_0= 0.25MWh$. We conduct 10 experiments with different true parameter sets, that are sampled independently from uniform distributions according to
\begin{equation}
    c_1 \sim U[0,20]\quad c_2 \sim U[0,20] \quad \eta \sim U[0.8,1.0]
\end{equation}
For each true parameter set, we take $T = 24$ with training data $N = 20$ and test data $N =10$ using hourly averaged real-time price data. 
The price data is randomly selected by dates from 2019 New York City (Zonn-J of New York Independent System Operator) real-time price data~\cite{nyiso}. We report the average performance of the proposed methods and all baseline methods by averaging across the 10 experiments.

For the implementation of the proposed gradient-based identification and forecast methods, we consider two scenarios: i) the system operator has knowledge about the parametric form of the objective function is quadratic thus using Algorithm \ref{alg:train} (Ours: quadratic); 
ii) the system operator does not know the cost function form thus using the generic identification Algorithm \ref{alg:train2} (Ours: generic). 
We use ICNN to model the generic utility function $u_\theta(\cdot)$. 
For the implementation of baseline methods. We compared the test loss with two baselines: a four-layer multi-layer perception with ReLU activations (MLP) and a recurrent neural network (RNN). 
Details about the network architecture and hyperparameters are provided in Appendix \ref{sec:model_detail}. %


\textbf{Forecast results.} 
Parameter identification results are shown in Table~\ref{tab:storage} and the average mean absolute error for parameter identification compared to the true parameters are shown in Table~\ref{tab:storage_iden}. In the learning process, $\overline{P}, \underline{P}$ can be inferred from the dispatch data $\{d^i, p^i\}_{i=1}^N$ by finding the maximum/minimum value of $d_t, p_t$. For the quadratic identification algorithm, we can directly identify all the parameters using the proposed Algorithm~\ref{alg:train}. 
For the generic identification algorithm, we use ICNN $u_{\theta}(d_t)$ to model the dis-utility function, and we estimate $c_1, c_2$ using the first-order Taylor approximation and the second-order Taylor approximation around the solution $d$,
\begin{equation*}
    c_1 = \frac{1}{T}\sum_t \nabla_{d_t} u_{\theta}(d_t), c_2 = \frac{1}{T}\sum_t \nabla^2_{d_t} u_{\theta}(d_t)
\end{equation*}
Fig~\ref{fig:energy_val2} shows the test loss comparison. The MSE test error of our method with quadratic model and generic model, MLP, and RNN among ten experiments are $\boldsymbol{8.74 \mathrm{e}{-5}}$, $\boldsymbol{7.37 \mathrm{e}{-4}}$, $0.018$ and $0.017$, respectively.
We observe that the quadratic identification model can recover the exact parameters of the ground-truth storage model, thus obtaining perfect forecasting results even with unseen price signals, as shown later in Fig \ref{fig:energy_pred}. The generic approach has a higher  error compared to the quadratic approach but still significantly outperforms other baseline methods.


\begin{table}[tb]
    \footnotesize
    \centering
    \caption{Model parameter identification using one parameter set. 
    }
    \begin{tabular}{c c c c c c}
    \hline 
         & $c_1$ & $c_2$ &  $\overline{E_m}$ & $\underline{E_m}$ & $\eta$ \\\hline
      True  & 11.02 & 14.16 & 0.5 & -0.5 & 0.9000 \\
      Ours (quadratic)  &  11.02 & 14.16 & 0.5 & -0.5 & 0.9000  \\
      Ours (generic) &  10.72 & 17.39 & 0.5 & -0.5 & 0.8990 \\
      \hline
    \end{tabular}
    \label{tab:storage}
\end{table}

\begin{table}[tb]
    \footnotesize
    \centering
    \caption{Average MAE of parameters identified compared to the true parameters in 10 experiments.
    }
    \begin{tabular}{c c c c c c}
    \hline 
         & $c_1$ & $c_2$ &  $\overline{E_m}$ & $\underline{E_m}$ & $\eta$ \\\hline
      Ours (quadratic)  &  0.39 & 0.20 & 0.007 & 0.004 & 0.007  \\
      Ours (generic) &  3.25 & 3.70 & 0.021 & 0.037 & 0.022 \\
      \hline
    \end{tabular}
    \label{tab:storage_iden}
\end{table}

The learned model then can be used to predict the energy storage dispatch with given price data.
We present the forecast results of Algorithm \ref{alg:train} (Ours: quadratic), Algorithm \ref{alg:train2} (Ours: generic),  compared with baselines including RNN and MLP in Fig~\ref{fig:energy_pred}. Though we do not have complete information on the model, we can forecast how the storage will schedule charging/discharging operations accurately, especially compared with black-box data-driven methods. The baseline RNN and MLP models take 10 times more samples to train, but they are not to accurately forecast the storage response.
\begin{figure}[tb]
    \centering
    \includegraphics[width=0.95\columnwidth]{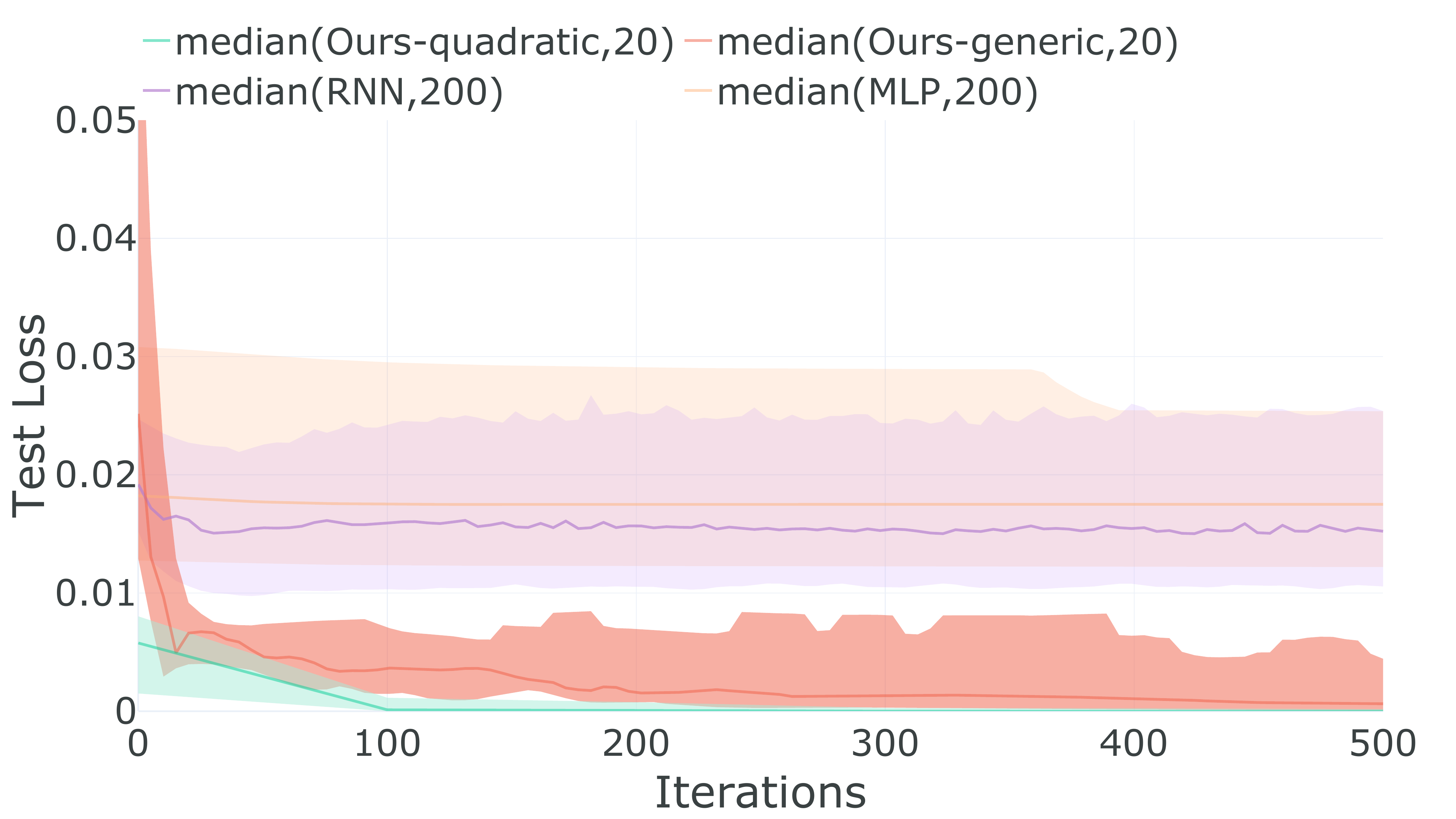}
    \caption{Comparison of test loss (MSE) versus iterations using our method and baselines, the solid line represents the median test loss of 10 experiments and shadow represents the 80\%/20\% quantile test loss of 10 experiments. Our method use 20 training samples and baselines use 200 training samples.}
    \label{fig:energy_val2}
\end{figure}
\begin{figure}[tb]
    \centering
    \includegraphics[width=0.95\columnwidth]{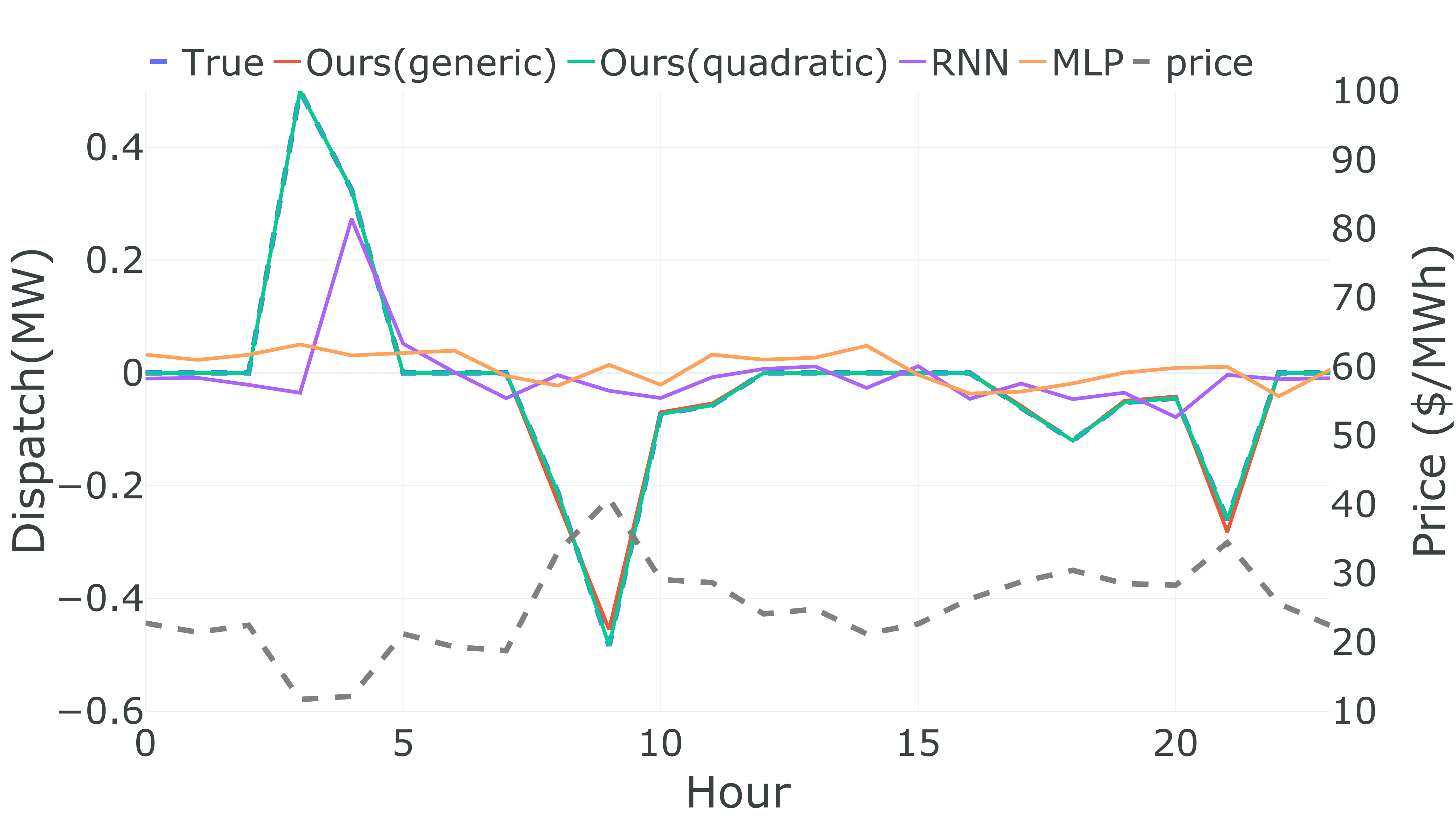}
    \caption{Comparison of energy storage dispatch using our methods and baselines (MLP and RNN). Our methods take 20 samples for training, while MLP and RNN take 200 samples for training. Figure shows a test example for one day. 
    Ours (quadratic) method identifies the true model and can predict the behavior exactly, thus the green line overlaps with the true response (blue).}
    \label{fig:energy_pred}
\end{figure}



\subsection{SoC-Dependent Energy Storage Model}
\textbf{Experiment setting.} 
Next, we consider storage arbitrage with a SoC-dependent degradation model in~\eqref{eq:agent_battery2}. We conduct 10 experiments with different true parameter sets. Same as the previous case, the storage power rating is set as $\overline{P} = 0.5MW, \underline{P} = 0MW$, starting with 50\% SoC, and vary the energy storage duration $H$ from $\{1h, 2h, 3h, 4h\}$ to calculate $\{\underline{E_m},\overline{E_m},e_0\}$. The ground-truth parameters are sampled independently from $c_1 \sim U[0,20], c_2 \sim U[0,20], \eta \sim U[0.8,1.0]$. We compare the performance of the quadratic storage identification algorithm (Algorithm \ref{alg:train}), the generic storage identification algorithm (Algorithm \ref{alg:train2}), and two data-driven baselines four-layer multi-layer perception with ReLU activations (MLP) and a recurrent neural
network (RNN). 
Details about the network architecture and hyperparameters are provided in Appendix \ref{sec:model_detail}.

\textbf{Convergence results}
As discussed in Section~\ref{sec:scp}, for generic agent models, we use the sequential convex programminng approach to sequentially approximate and solve the convex quadratic program. 
Once we reached the fixed point, we use backpropagation to compute the derivatives. 
Figure~\ref{fig:convergence} (left) shows the objective $c^{ite}(y^{(0)})$ for different iterations during the training process. It shows that the training loss keeps reducing as training iterations increase. In Figure~\ref{fig:convergence} (right), we show the convergence of the sequential convex approximation process, i.e., $c^{ite}(y^{(k)}) = \sum_{t=1}^T \lambda_t(p^{(k)}_t-d^{(k)}_t)+u_{\theta}(d^{(k)})$ for $k = 0, 1, \ldots, 19$ during each training iteration $it$, where $y^{(k)}$ is the solution to the approximated convex problem from iteration $k-1$. As parameters are updated during training, cost values are in different scales, thus we show the normalized $\frac{c^{ite}(y^{(k)})}{c^{ite}(y^{(0)})}$ for different iterations $ite$. From the figure, we see the fixed point can be obtained within 20 iterations.
\begin{figure}[tb]
    \centering
    \includegraphics[width=0.95\columnwidth]{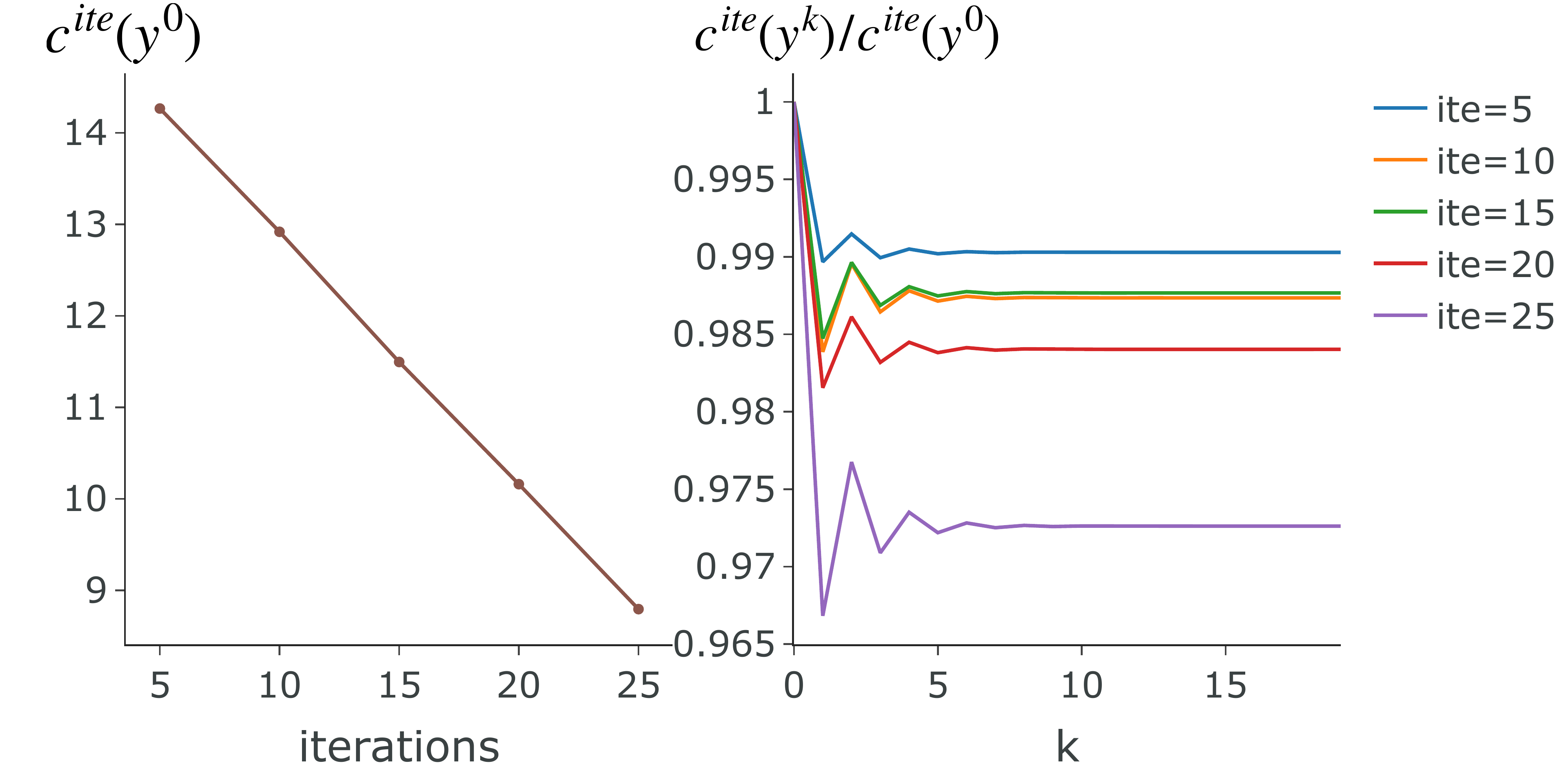}
    \caption{(Left) Cost objective during the training process. (Right) Convergence of the sequential convex programming reaches for 20 iterations.}
    \label{fig:convergence}
\end{figure}

\textbf{Forecast results.} 
Now we show the capability of predicting energy storage dispatch using different methods. We use $N=40$ training samples for our method, and $N=200$ training samples for baselines (MLP and RNN) since performance of these methods degrade significantly with fewer data. The test results are based on 10 days of unseen prices. The average test MSE loss of our method with Algorithm \ref{alg:train} (Ours: quadratic), Algorithm \ref{alg:train2} (Ours: generic), MLP and RNN among ten experiments are $0.032$, $\boldsymbol{0.016}$, $0.039$ and $0.039$. Our approach based on generic model Algorithm 2 can accurately predict the storage response, while the baseline data-driven models such as MLP and RNN are barely able to capture the response. Since the storage model is non-quadratic, the quadratic identification method Algorithm 1 is subject to a model mismatch and we observe a performance degradation compared to the previous case, while the testing error is lower than the pure data-driven baselines. This again highlights the importance of incorporating physical knowledge into black-box data-driven models for energy storage agent behavior forecasting.
In addition, we demonstrate the forecasting performance in one randomly selected test day, shown in Fig~\ref{fig:energy_pred2}.

\begin{figure}[tb]
    \centering
    \includegraphics[width=0.95\columnwidth]{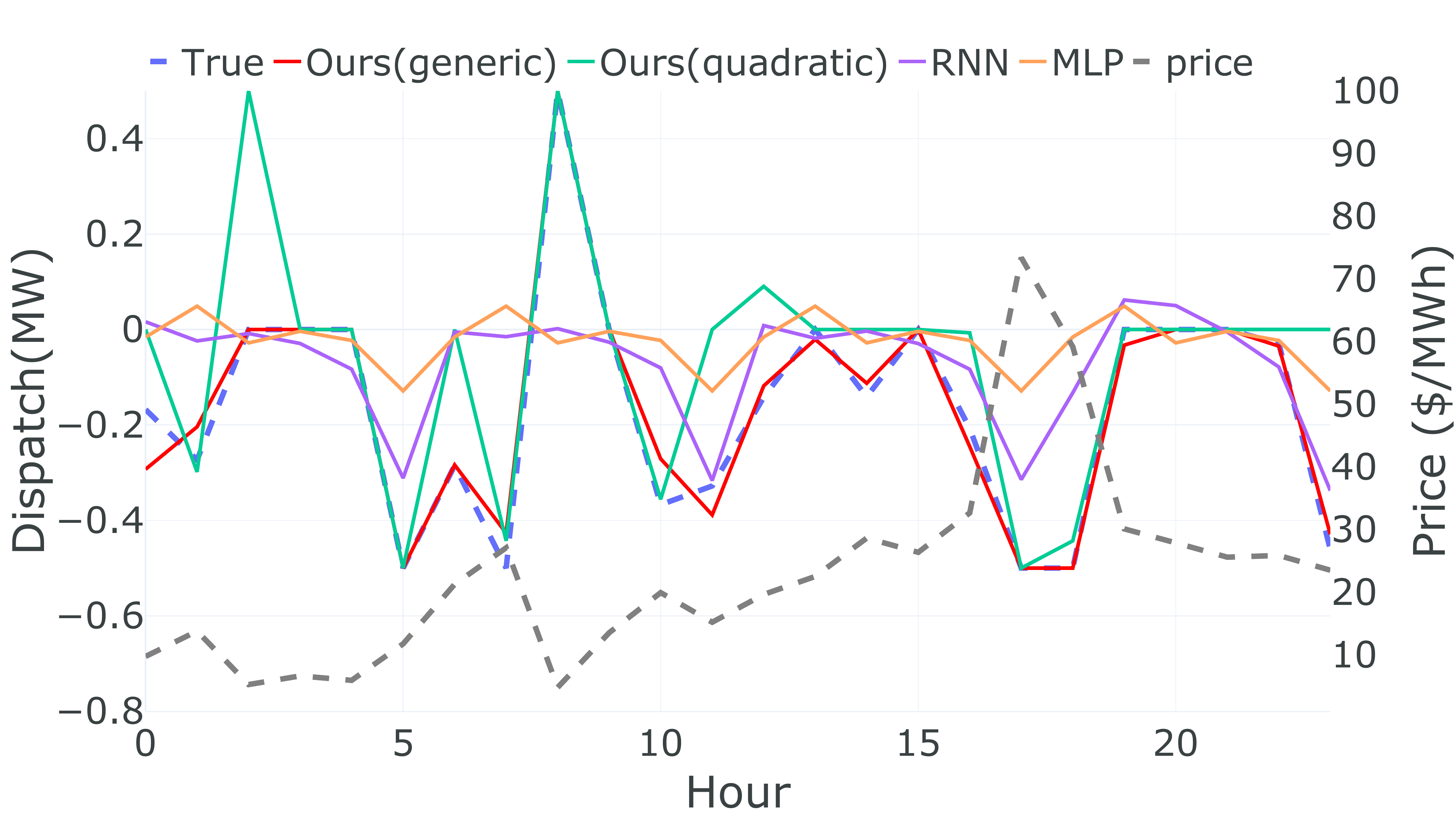}
    \caption{Comparison of energy storage dispatch using our methods and baselines (MLP and RNN). Figure shows a test example for one day. Positive values represent charging and negative values represent discharging. }
    \label{fig:energy_pred2}
\end{figure}


\subsection{Real-world Dataset: Queensland Battery Project}
We now test the performance of the proposed approach using a real-world storage arbitrage dataset from the University of Queensland Australia (UQ). UQ installed a 1.1~MW/2.22~MWh developed using Tesla Powerpack 2.5 to provide price arbitrage and operating reserves.
Contingency frequency events typically occur only a small number of times each year~\cite{UQreport}. Thus the storage charge and discharge behaviors are dominated by price arbitrage in real-time markets.

\textbf{Experiment setting.} We use data from 2020,
the original data consists of recordings of 5 minutes resolution between January 2020 and June 2020. We split the dataset into sequence samples lasting 40 hours with half hourly resolution (i.e. $T=80$) as the UQ project uses a model predictive control approach with up to 40 hours look-ahead~\cite{UQreport}.
We drop all the samples that include contingency frequency regulation events. We also notice that demand response engine does not perform charging/discharging in some dates, due to the manual intervention. We also drop those samples and the final dataset contains 42 samples (1680 hours). We use 22 samples for training and 20 samples for testing. We use exact storage capacity value from the project description with 
$\overline{P}=1.1 MW, \underline{P}=0 MW, \overline{E_m} = 2.22 MWh, \underline{E_m}=0 MWh $. 
\revise{We compare the performance of the quadratic
storage identification algorithm (Algorithm 1) and the generic
storage identification algorithm (Algorithm 2) with baselines including MLP, RNN and a threshold-based method. For threshold-based method, we use the training set to record the activated price ratio $r^p_i = \frac{\lambda_{it}^p}{\overline{\lambda_i}}, r^d_i = \frac{\lambda_{it}^d}{\overline{\lambda_i}}$, where $\overline{\lambda_i}$ is the average price of $i$-th sample, and $\lambda_{it}^p, \lambda_{it}^d$ is the activated price to charge or discharge for $i$-th sample at time $t$. We use the median $r^p$ and $r^d$ as the activated ratio for both the training set and test set. }

\textbf{Forecast results.}
\begin{figure}[tb]
    \centering
    \includegraphics[width=0.9\columnwidth]{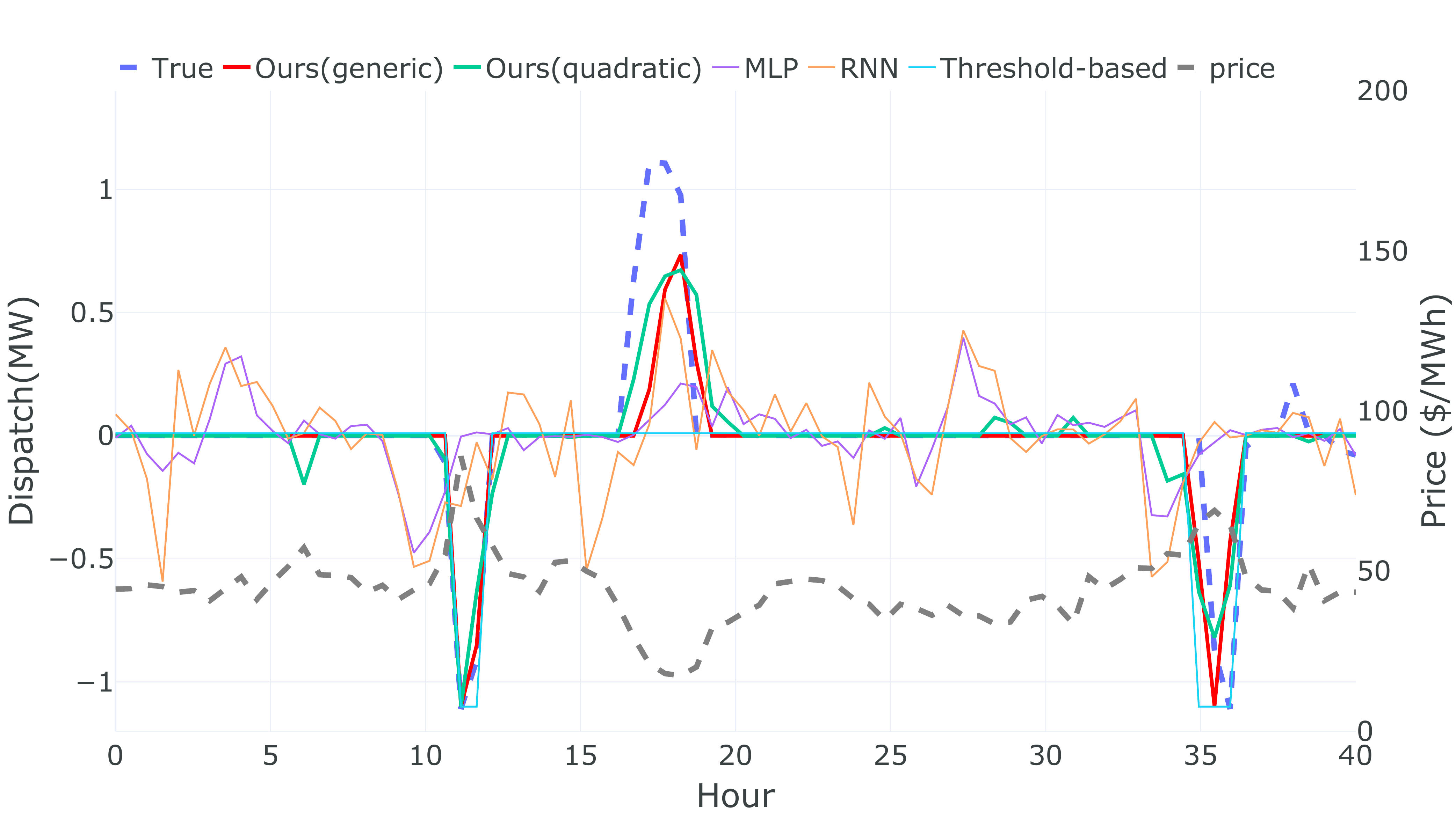}
     \includegraphics[width=0.9\columnwidth]{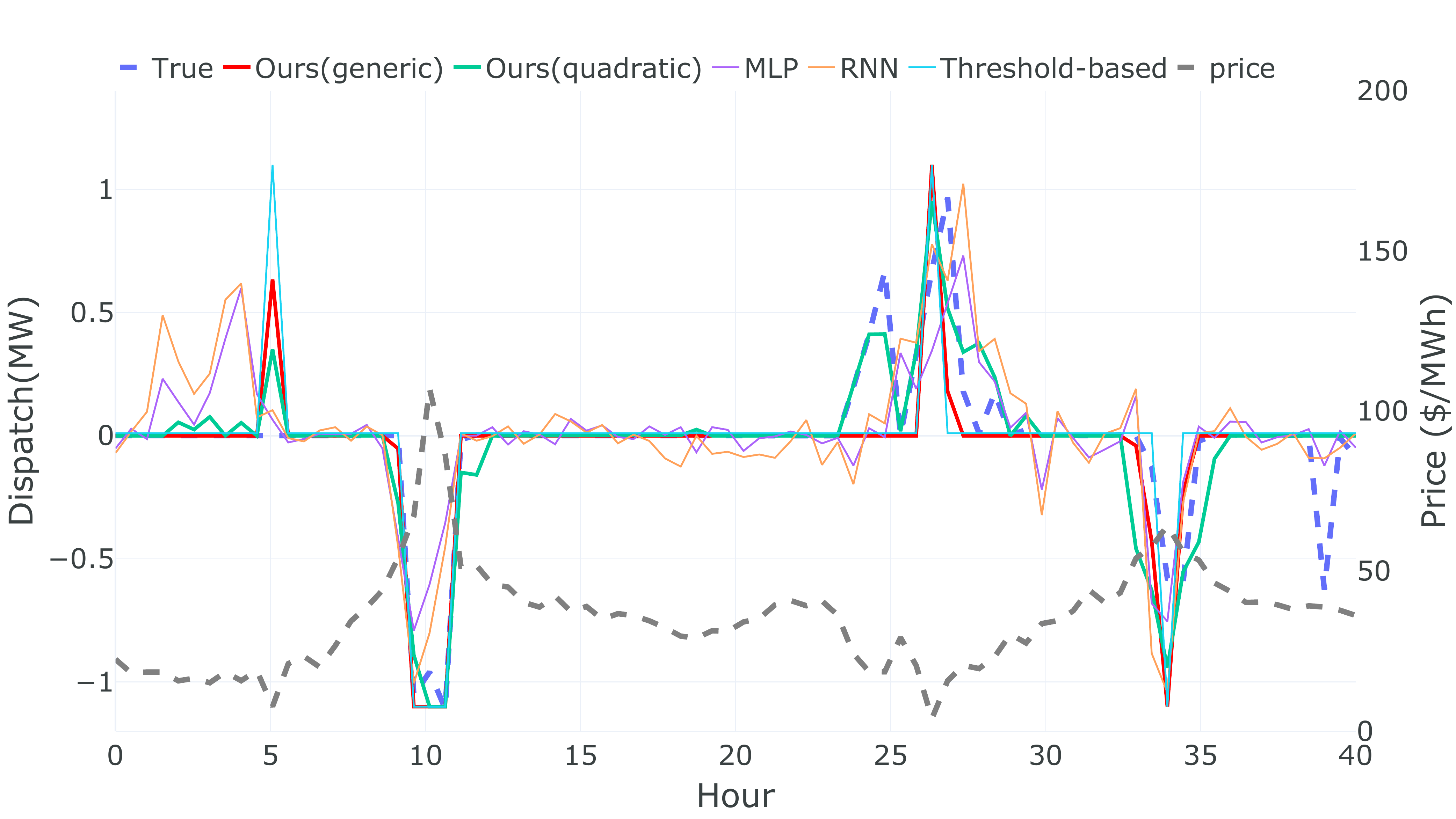}
    \caption{\revise{Comparison of energy storage dispatch using our methods and baselines (RNN, MLP, Threshold-based). The figure shows two examples of 40 hours' length. Positive values represent charging and negative values represent discharging.}}
    \label{fig:tesla}
    \vspace{-6pt}
\end{figure}
Fig~\ref{fig:tesla} shows the forecast results on two test samples. Table~\ref{tab:real} shows training and testing MSE and correlation scores between the forecasted storage behavior and the actual storage behavior in test dataset. We compute correlation coefficient of sequence $x$ and $y$ using $$corr(x,y) = \frac{\sum (x_i-\bar{x})(y_i-\bar{y})}{\sqrt{\sum(x_i-\bar{x})^2\sum (y_i-\bar{y})^2}}$$
\revise{Our experiment results demonstrate that both of our proposed methods outperform the baselines in terms of test MSE and correlation. Figure \ref{fig:tesla} shows that the MLP and RNN baselines (purple and orange lines) are unable to capture future storage behaviors (blue dashed line), while the threshold-based method is better but fails to accurately predict the dispatch behavior driven by optimization. Our proposed methods (red and green lines) generally perform well in capturing future storage behaviors (blue line). However, the quadratic model provides better testing accuracy than the generic model. This may be due to the fact that the University of Queensland project uses a quadratic programming algorithm to optimize storage operation, making the quadratic identification model a better fit. The UQ project report~\cite{UQreport} acknowledges mismatches between actual and simulated optimal battery responses, as well as differences between the actual and simulated prices used in the optimization. Despite these challenges, our data-driven approach can still accurately forecast energy storage behaviors and outperforms existing approaches.}
\vspace{-6pt}
\begin{table}[htbp]
    \footnotesize
    \centering
    \caption{\revise{Mean square error (MSE) and correlation score.}}
    \begin{tabular}{c c c c}
    \hline 
      & \revise{Training MSE} & \revise{Test MSE} & \revise{Test Correlation}\\
     \hline
     \revise{Ours(quadratic)} & \revise{\textbf{0.044}} & \revise{\textbf{0.042}} &  \revise{\textbf{0.74}}\\
     \revise{Ours(generic)} & \revise{0.054} &   \revise{0.053} &  \revise{0.68} \\
    \revise{Threshold-based} &  \revise{0.088} &  \revise{0.070} & \revise{0.57} \\
     \revise{MLP} & \revise{0.047} & \revise{0.096} & \revise{0.34} \\
     \revise{RNN} & \revise{0.092} & \revise{0.109} & \revise{0.27} \\
      \hline
    \end{tabular}
    \label{tab:real}
\end{table}
\vspace{-6pt}
\revise{\subsection{Gradient-based approach v.s. optimization tools for solving the bi-level optimization}}
\revise{
To solve the bi-level optimization in~\eqref{eq:identification}, an alternative approach is using the KKT conditions for the low-level problem~\eqref{eq:agent_obj}~\eqref{eq:agent_battery_contr1}, \eqref{eq:storage_contr2_reformulation} and therefore we can formulate a single-level optimization and use optimization tools to solve it.}
\revise{Table~\ref{tab:compare_with_opt} shows the comparison of our method and optimization tool on the quadratic storage model. All the experiments are conducted at CPU Intel(R) Core(TM) i7-9750H CPU @ 2.60GHz. We compare the nonlinear Gurobi solver~\cite{gurobi} and the Baron solver~\cite{ryoo1995global}. As Gurobi demonstrated better performance, we present the results obtained with Gurobi solver. The code is also available at \href{https://github.com/alwaysbyx/Predicting-Strategic-Energy-Storage-Behaviors/blob/main/quadraticenergystorage/optimization.m}{Github Link}. The optimization tool is only able to handle small datasets and is unable to find a solution for larger datasets (N=20) within a reasonable timeframe (1 hour). In contrast, the proposed gradient method is able to handle larger datasets and achieves better test accuracy when dealing with the same amount of data. The results demonstrate that the proposed gradient-based method outperforms the optimization method in terms of both computation efficiency and prediction accuracy. Additionally, Table~\ref{tab:compare_robust} demonstrates that the gradient-based method maintains its performance across different levels of noise, while the optimization approach exhibits a significant decline in performance as the noise increases.}


\begin{table}[htbp]
\caption{\revise{Comparison of Gradient-based method(gradient) and Single-level optimization(optimization): Time and Test mean-square error (MSE). }}
\centering
\begin{tabular}{ccccc}\hline
Training size              & Method       & $N=1$               & $N=5$                 & $N=20$                \\\hline
\multirow{2}{*}{Time(s)}   & Gradient     & 28.72  & 30.24    & 132.4    \\
                           & Optimization & 600            & 3600         & -                 \\\hline
\multirow{2}{*}{Test MSE } & Gradient     & 4.4$\times 10^{-4}$ & 1.8$\times 10^{-4}$   & 3.6$\times 10^{-5}$   \\
& Optimization & 48.4$\times 10^{-4}$          & 210.2$\times 10^{-4}$ & - \\\hline
\label{tab:compare_with_opt}
\end{tabular}
\end{table}
\begin{table}[htbp]
\caption{\revise{Comparison of Gradient-based method(gradient) and Single-level optimization(optimization) with different level of noise using $N=1$ training data, $\mathcal{N}$ is Gaussian distribution.}}
\centering
\begin{tabular}{ccp{1cm}p{1cm}p{1cm}}\hline
Training size              & Method       & No noise               & Noise $\sim \mathcal{N}(0,0.001)$              & Noise $\sim \mathcal{N}(0, 0.05)$                \\\hline
\multirow{2}{*}{Test MSE} & Gradient     & 4.4$\times 10^{-4}$ & 4.5$\times 10^{-4}$   & 4.8$\times 10^{-4}$   \\
& Optimization & 48.4$\times 10^{-4}$                & 87.8$\times 10^{-4}$  & 192.1$\times 10^{-4}$  \\\hline
\label{tab:compare_robust}
\end{tabular}
\end{table}

\section{Conclusions and Future Works}\label{Conclu}
In this paper, we proposed a novel gradient-based approach that combines data-driven methods and model-based optimization for the identification and prediction of strategic energy storage behaviors. 
We provide three case studies to validate the performance of the proposed approach: a storage model with a convex quadratic function, a storage model with a generic SoC-dependent degradation function, and a real-world energy storage demand response engine from the University of Queensland battery project. The experiment results show our method can identify the model well and predict the strategic storage behavior accurately, which significantly outperforms state-of-the-art data-driven baselines. \revise{The proposed algorithm can be applied to demand response behaviors forecast for building energy systems~\cite{fernandez2021forecasting,tang2019model} aggregated demand response participants~\cite{yuan2022data, tan2023data} and other potential distributed energy resources that can be formulated as optimization problems with affine constraints and generic convex cost functions. 
Interesting future directions include extending the proposed
algorithm for behavior modeling and forecasting of multiple strategic storage units that are price-makers.  Another interesting future direction is market tariff design based on the identified agent behavior model.}


\bibliographystyle{IEEEtran}
\bibliography{reference}

\appendix 
\subsection{Proof}\label{sec:proof}
\begin{proof}[Proof of Theorem \ref{convergent_for_eq}]\label{proof_dominance}
Let $h(x) = \frac{1}{x}$, by assumption $\alpha > \delta > 0$, thus $h(\alpha) = \frac{1}{\alpha} \in (0,\frac{1}{\delta})$. Let $q = \begin{bmatrix} x, b \end{bmatrix}^\tr$, and 
\begin{equation}
    L(\alpha, b) = \frac{1}{N}\sum_{i=1}^N \left\|\frac{1}{\alpha}k_{1i} + k_2 b - y_i\right\|_2^2, 
\end{equation}
\begin{equation}
    g(q) = g(x, b) := L(h(x), b) = \frac{1}{N}\sum_{i=1}^N \|xk_{1i} + k_2 b -y_i\|_2^2. 
\end{equation}

Since $g$ is strongly convex, there exists $\mu > 0$ such that,
\begin{equation}\label{eq:strong_convex}
    g(q') \geq g(q) + \nabla g(q)^\top (q'-q) + \frac{\mu}{2} \|q' - q\|^2\,,
\end{equation}
for all $q', q$. Taking minimization respect to $q'$ on both sides of Eq~\eqref{eq:strong_convex} yields
$g(q^\star) \geq g(q) - \frac{1}{2\mu} \|\nabla g(q)\|^2\,,$
where $q^\star$ denotes the unique global minimal point of $g$.  
Re-arranging it we have,
$2\mu(g(q) - g(q^\star)) \leq \|\nabla g(q)\|_2^2, \quad \forall q \in \mathbb{R}^2$
where $\mu$ is the strong-convexity constant of $g$. 

As $h$ is a continuous and one-to-one mapping, we have $L(\alpha, b) = g(x,b), \forall q$.
Let $\J_h: \R \rightarrow \R$ denote the Jacobian function of $h(\cdot)$. Applying the derivative chain-rule leads to
$$
\begin{aligned}
2\mu(L(\alpha,b)-L(\alpha^\star,b^\star)) &\leq \|\nabla g(h^{-1}(\alpha),b)\|_2^2 \\
& = \|\J_h(h^{-1}(\alpha))\nabla L(h(h^{-1}(\alpha)),b)\|_2^2 \\
& \leq \|\J_h(h^{-1}(\alpha))\|_2^2\|\nabla L(\alpha,b)\|_2^2 \\
& \leq \|\J_h(h^{-1}(\alpha))\|_F^2\|\nabla L(\alpha,b)\|_2^2, \forall \alpha,b.
\end{aligned}
$$

Observe that Jacobian $\J_h(h^{-1}(\alpha))$ is a polynomial for every $\alpha \in \R$. Then $\J_h(h^{-1}(\alpha))$ is bounded on any compact set. As shown in~\cite{furieri2020learning}, the sublevel set $\G_{10\epsilon^{-1}} = \{\alpha,b \in \R | L(\alpha,b) - L(\alpha^\star,b^\star) \leq 10\epsilon^{-1}\Delta_0\}$ is compact because $h$ is a continuous and bounded map, and $g$ is strongly convex. We denote
$$\tau = \sup_{\alpha\in \G_{10\epsilon^{-1}}} \|\J_h(h^{-1}(\alpha))\|^2_F$$
which is bounded constant. By setting $\mu_\epsilon = \frac{2\mu}{\tau}$, we obtain
\begin{equation}\label{eq:PL_inequality_alpha}
    \mu_\epsilon(L(\alpha,b)-L(\alpha^\star,b^\star)) \leq \|\nabla L(\alpha,b)\|_2^2, \forall \alpha,b \in \G_{10\epsilon^{-1}}
\end{equation}
where $\mu_{\epsilon} > 0$ is the gradient dominance constant over $\G_{10\epsilon^{-1}}$, for any $\epsilon > 0$. This inequality shows that the gradient grows faster than a quadratic function when $L(\alpha,b)$ moves away from the optimal function value. In addition, 
$$\nabla_{\alpha} L(\alpha,b) = \frac{1}{N}\sum_{i=1}^{N} -\frac{2}{\alpha^3} k_{1, i}^\top k_{1, i} - \frac{2}{\alpha^2} k_{1, i}^\top k_{2} b + \frac{2}{\alpha^2} k_{1, i}^\top y_i$$
$$\nabla_{b} L(\alpha,b) = \frac{1}{N}\sum_{i=1}^N (\frac{2}{\alpha}k_{1i} + bk_2 - y_i)^\tr k_2$$
has a Lipschitz constant $L$ and
$L$ is finite since $\G_{10\epsilon^{-1}}$ is compact. With the gradient dominance property in \eqref{eq:PL_inequality_alpha}, and the L-Lipschitz continuity of $\nabla L(\alpha,b)$, the following gradient method with a step-size of $1/L$,
$$\alpha^{(t+1)} = \alpha^{(t)} - \frac{1}{L} \nabla_\alpha L(\alpha^{(t)},b ^{(t)})\,,
b^{(t+1)} = b^{(t)} - \frac{1}{L} \nabla_b L(\alpha^{(t)}, b^{(t)})\,,$$
has a global linear convergence rate over $\G_{10\epsilon^{-1}}$,
$$L(\alpha^{(t)},b^{(t)}) - L(\alpha^\star,b^\star) \leq (1-\frac{\mu_{\epsilon}}{2L})^t (L(\alpha_0,b_0) - L(\alpha^\star,b^\star)).$$
The linear convergence of gradient descent under gradient dominance properties was first proved by \cite{polyak1963gradient} and recently used in various of nonconvex optimization problems~\cite{karimi2016linear,furieri2020learning}.
In addition, since $g(q)$ admits an unique optimizer $q^\star$, and correspondingly, $L(\alpha,b)$ has an unique minimizer $\alpha^\star,b^\star$. Thus $\lim_{t \rightarrow \infty}\alpha^{(t)} \rightarrow \alpha^\star, \lim_{t \rightarrow \infty}b^{(t)} \rightarrow b^\star$.
\end{proof}

\subsection{Model Details}\label{sec:model_detail}
All experiments: We use Adam optimizer with learning rate $0.01$.
\begin{outline}
 \1 Batch size: 32 for MLP and RNN; data size for quadratic model and generic model.
 \1 Baseline Models
     \2 \underline{MLP}: four layers with 64 hidden units in each layer.
     \2 \underline{RNN}: four LSTM layers with one fully-connected layer.
 \1 Architecture of input convex neural network (ICNN): As Fig~\ref{fig:icnn} shows, we use 
 two hidden layers with 24 hidden units, softplus~\cite{zheng2015improving} as an activation function.
\end{outline}



\begin{figure}[h]
    \centering
    \includegraphics[width=0.65\columnwidth]{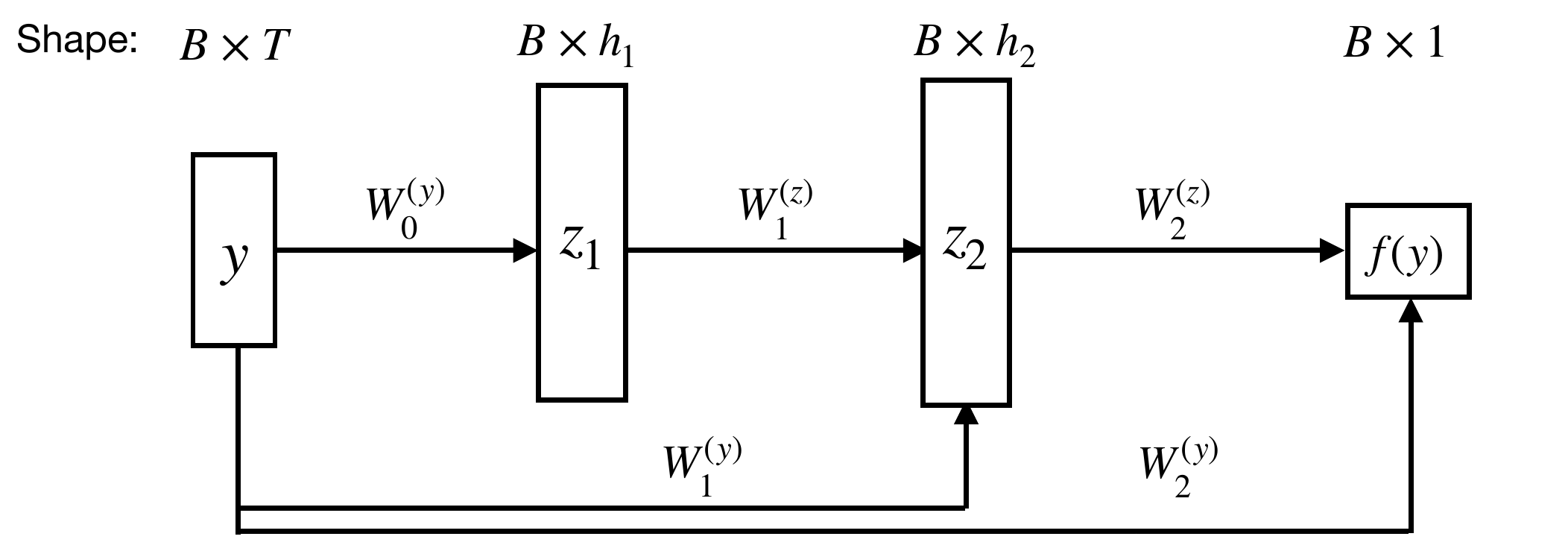}
    \caption{Architecture of ICNN. $B$ denotes the batch size, and $T$ is the input dimension, the output is a scalar where the dimension is 1.}
    \label{fig:icnn}
\end{figure}

\end{document}